\documentclass{article}

\usepackage[utf8]{inputenc}
\usepackage[usenames,dvipsnames]{xcolor} 
\usepackage[flushleft]{threeparttable}
\usepackage[us,12hr]{datetime}
\usepackage[titletoc,title]{appendix}

\usepackage{graphicx} 
\usepackage{natbib}
\usepackage{amsmath}
\usepackage{amssymb}
\usepackage{amsthm} 
\usepackage{dsfont}
\usepackage{tablefootnote}
\usepackage{fancyhdr}
\usepackage{multirow, xltabular}
\usepackage{multicol} 
\usepackage{lipsum}
\usepackage{float}
\usepackage{comment}
\usepackage{subfig}
\usepackage{epsfig}
\usepackage{epstopdf} 
\usepackage{fancyvrb}
\usepackage{tabularx}
\usepackage{caption} 
\usepackage{hyperref}

\newtheorem{theorem}{Theorem}

\begin{document}

\title{Statistical Properties and Power Analysis of Divergence Measures for Credit Risk Model Monitoring}

\author{
  Abdullah Karasan\textsuperscript{1,*} \quad Alper Hekimoğlu\textsuperscript{2} \\[10pt]
  \small\textsuperscript{1}University of Maryland, Department of Computer Science and\\
  \small Electrical Engineering, Baltimore, Maryland, USA \\
  \small\textsuperscript{2}European Investment Bank, Model Validation, Luxemburg, Luxemburg \\[6pt]
  \small\textsuperscript{*}Corresponding author: akarasan@umbc.edu
}
\date{}
\maketitle

\begin{abstract}
Divergence measures are essential tools for detecting distributional shifts in model monitoring, particularly crucial given the volatility of financial data. While the Population Stability Index is most widely used, Jensen–Shannon Divergence and Kullback–Leibler Divergence offer distinct advantages: Jensen–Shannon Divergence handles mixture models, addresses zero-binning problems, and is symmetric, while Kullback–Leibler Divergence excels in Bayesian model comparison.

This study extends the work of \cite{yurdakulNaranjo} with two primary contributions: First, we derive the statistical properties and chi-square benchmark values for Jensen–Shannon Divergence and Kullback–Leibler Divergence. Second, we demonstrate their applicability by detecting distributional changes in credit default probabilities from Merton, Merton with jump, and stochastic volatility with jump models.

Our results establish that Jensen–Shannon Divergence and Kullback–Leibler Divergence follow chi-square distributions and reveal important practical trade-offs. Jensen–Shannon Divergence exhibits superior Type I error control, maintaining rejection rates closest to 5\%, thereby minimizing false positives. However, this conservatism reduces statistical power at small samples (27\% vs. 32\% for Population Stability Index and Kullback–Leibler Divergence at n=m=200), requiring larger samples for reliable detection. This trade-off enables practitioners to select measures based on whether minimizing false alarms or maximizing detection sensitivity is prioritized.

\textbf{Keywords}: Distribution shift, Divergence, Credit Risk, Simulation, Machine Learning
\end{abstract}

\section{Introduction}
Modeling is a challenging and multi-dimensional process, which is considered incomplete without robust and efficient model monitoring. This challenge lies in determining when and how to apply this monitoring activity. This is where distribution shift becomes important, as its detection in the monitoring process is essential in ensuring the long-term applicability of the model. Therefore, this shift should be monitored periodically to allow for re-training of the model to adjust to the current distribution. A model without a proper monitoring mechanism may suffer in data drift, which can significantly degrade the performance of the deployed model and, in some cases, make model useless. In this sense, divergence measures serve as strong tools for detecting distribution shift problems, which are indispensable in information theory and have recently become primary tools in numerous fields.

In the literature, there are several divergence measures that are used to distinguish one probability distribution from another. The most popular and useful measures are the Population Stability Index (PSI), Kullback-Leibler Divergence (KLD), and Jensen-Shannon Divergence (JSD). However, we have found that the derivation of the distributions and benchmark values for KLD and JSD have not been studied, despite their high importance and frequent use.

Given the volatile nature of financial data, it is of quite importance to consider the distribution shift in unseen data. In finance, these divergence measures enhance model validation, calibration, and portfolio monitoring. Their role in credit scoring strengthens models against changing economic landscapes. As early warning indicators, these measures unveil emerging risks, signaling adaptability in risk management. Thus, use of divergence measures empowers credit risk models for agility, fairness, and resilience in risk assessment.

In this study, we obtain the statistical properties of JSD and KLD measures and show that these measures follow a $\chi^2$ distribution. We also determine their benchmark values, which can be used to assess the statistical significance level. In the empirical section, we present the behavior of PSI, KLD, and JSD measures, as well as the differences between the empirical distributions obtained from credit risk models. Additionally, we compare and evaluate the consistency of these divergence measures through statistical power analysis. In order to accomplish these tasks, we use well-known credit risk models, including Merton, Merton with jumps, and Stochastic Volatility with jumps. The underlying assets in this study are latent asset prices. We use these latent asset prices to estimate the probability of default, which are then used in statistical analyses to detect the divergence, if any. We simulate multi-dimensional stock prices, where the correlations between these simulated stocks greatly affected the probability of default, as well as the statistical analyses based on it. This study contributes to the literature in two main ways. Firstly, we derive the statistical properties and benchmark values for JSD and KLD measures and show these measures follow $\chi^2$ distribution. Secondly, we demonstrate the way it can be applicable to financial models. This allows us to assess which divergence measure performs best under certain conditions.

The remainder of this study is as follows. In the second part, the statistical analysis of divergence measures are studied. In the third part, the aforementioned structural credit risk models are presented and discussed. Fourth part is dedicated to the empirical application including distributional comparison of divergence measures as well as statistical power analysis. Part five concludes.

\section{Literature Review}
The robustness of credit risk models relies heavily on the stability of the underlying data distributions. In the context of model monitoring, divergence measures are standard tools used to quantify distributional shifts over time. The foundations of categorical data analysis, which often underpin the binning approaches used in these assessments, were comprehensively established. Building on these foundations, the Population Stability Index (PSI) has become widely utilized in the banking industry. Recently, \cite{yurdakulNaranjo} detailed the formal statistical properties of the PSI, providing a rigorous mathematical framework that allows risk practitioners to apply formal hypothesis testing rather than relying on heuristic thresholds.

While PSI is an industry standard, information-theoretic measures such as the Kullback-Leibler Divergence (KLD) and the Jensen-Shannon Divergence (JSD) offer significant mathematical advantages. JSD, as a symmetrized and smoothed version of KLD, addresses the inherent asymmetry of KLD. \cite{Endres2006} introduced a new metric for probability distributions, mathematically proving that the square root of JSD satisfies the triangle inequality and functions as a true distance metric. The mathematical properties of classical and quantum Jensen-Shannon divergence were further explored by \cite{briet2009properties}, demonstrating its versatility. More recently, \cite{Nielsen2020} generalized the Jensen-Shannon Divergence and explored its centroid computations, expanding its theoretical scope in information geometry.

In financial applications, structural credit risk models must often account for sudden, discontinuous market movements. \cite{tankov} provided a comprehensive framework for financial modeling with jump processes, detailing the mechanics of models such as the Merton model with jumps and Stochastic Volatility with Jumps (SVJ). As these complex structural models generate nuanced probability of default (PD) distributions, standardizing the detection of shifts in these distributions is critical.

Despite the extensive individual literature on probability metrics \citep{Endres2006, briet2009properties, Nielsen2020} and jump diffusion processes \citep{tankov}, there is a distinct gap in linking the two for practical risk management. Specifically, the exact asymptotic distributions of KLD and JSD for sample-based model validation have not been derived in the same rigorous manner as PSI \citep{yurdakulNaranjo}. This study bridges that gap by deriving the exact statistical properties of JSD and KLD, and subsequently applying them to detect distributional shifts in the advanced structural credit risk models outlined in the literature.

\section{Statistical Analysis of Divergence Measures}\label{sec:2}
In the context of credit risk literature, PSI is often utilized to calculate the divergence between two credit risk models. This divergence measure is also frequently used to detect shifts in the distribution of model features used in predicting default probability. Additionally, we consider two well-known measures, JSD and KLD, and propose using JSD in particular as an alternative stability measure for predicting the probability of default. JSD is a symmetrized and smoothed form of KLD in information theory \cite{briet2009properties}. Moreover, JSD is a metric, and it is the square of a metric as proven by Endres and Schindelin \cite{Endres2006}. This function has several striking features: it is defined everywhere, bounded, symmetric, and vanishes only when $p = q$ (where $p$ and $q$ are two distributions). Furthermore, JSD can be applied to densities with any support and is bounded by $\log(2)$ \cite{Nielsen2020}.
 
On the other hand, KLD is not a metric. Besides, KLD is asymmetric as the cross-entropy itself is asymmetric. More specifically, let $p_1$ and $p_2$ be two probability distributions and KLD is asymmetric:
\begin{equation}
D_{KL}(p_1||p_2)\neq D_{KL}(p_2||p_1)
\end{equation}
 
Thus, KLD is not a distance measure as it has a non-negative value indicating the extent to which two probability distributions diverge. Thanks to the absolute continuity constraint, this asymmetry property of KLD is not as bad as it first seems. KLD is, therefore, a measure on detecting the difference of two probability distributions.
 
Before moving on, it is worth defining the divergence measures we discussed before. Let $N$ be sample sizes for base and target population. Similarly, let $\hat{p}$ and $\hat{q}$ are relative frequencies of the base and target periods.
\begin{align}
  \mathrm{PSI} &= \sum_{i=1}^{B}(\hat{p}_i-\hat{q}_i)\bigl(\log\hat{p}_i-\log\hat{q}_i\bigr),\label{eq:PSI}\\
  \mathrm{KLD} &= \sum_{i=1}^{B}\hat{p}_i\bigl(\log\hat{p}_i-\log\hat{q}_i\bigr),\label{eq:KLD}\\
  \mathrm{JSD} &= \tfrac{1}{2}\sum_{i=1}^{B}
                 \Bigl[\hat{p}_i\bigl(\log\hat{p}_i-\log\hat{Q}_i\bigr)
                      +\hat{q}_i\bigl(\log\hat{q}_i-\log\hat{Q}_i\bigr)\Bigr].\label{eq:JSD}
\end{align}
where $\hat{Q}_{i}=\frac{1}{2}(\hat{p}_{i}+\hat{q}_{i})$.

\begin{theorem}[General quadratic approximation of JSD]\label{theorem:th1}
Let $\hat{p}_i = X_i/n$ and $\hat{q}_i = Y_i/m$, where
$(X_1,\dots,X_B)\sim\mathrm{Multinomial}(n;\,p_1,\dots,p_B)$ and
$(Y_1,\dots,Y_B)\sim\mathrm{Multinomial}(m;\,q_1,\dots,q_B)$. Then
\begin{equation}\label{eq:JSD_approx}
  \mathrm{JSD}
  \;\approx\;
  \frac{1}{2}\sum_{i=1}^{B}
  \left[
    \frac{(\hat{p}_i-\hat{q}_i)^{2}}{2p_i}
    -\frac{(\hat{p}_i+\hat{q}_i)^{2}}{4p_i}
    +p_i
  \right].
\end{equation}
\end{theorem}

\begin{proof}
We expand $\log\hat{p}_i$, $\log\hat{q}_i$, and $\log\hat{Q}_i$ in Taylor series about
the population values $p_i$, $q_i$, and $Q_i=\tfrac{1}{2}(p_i+q_i)$, retaining terms
up to second order. Substituting these expansions into~\eqref{eq:JSD} and using
$(\hat{p}_i-p_i)^2=\mathcal{O}(1/n)$ and $(\hat{q}_i-q_i)^2=\mathcal{O}(1/m)$ to
discard higher-order remainders, followed by the asymptotic convergence
$\hat{p}_i\xrightarrow{p}p_i$ and $\hat{q}_i\xrightarrow{p}q_i$, yields~\eqref{eq:JSD_approx}.
Full algebraic details are given in Appendix~\ref{app:JSD}.
\end{proof}

\begin{theorem}[JSD under the null]\label{theorem:th1b}
Under the conditions of Theorem~\ref{theorem:th1}, if additionally $p_i=q_i$ for all $i$, then
\begin{equation}\label{eq:JSD_null}
  \mathrm{JSD}
  \;\approx\;
  \frac{1}{8}\sum_{i=1}^{B}\frac{(\hat{p}_i-\hat{q}_i)^{2}}{p_i}
  \;\sim\;
  \frac{1}{8}\,\chi^{2}_{B-1}\!\left(\frac{1}{m}+\frac{1}{n}\right).
\end{equation}
\end{theorem}

\begin{proof}
Setting $p_i=q_i$ in~\eqref{eq:JSD_approx} and simplifying shows that the terms
$-\hat{p}_i\hat{q}_i/p_i+p_i$ vanish asymptotically, leaving
$\mathrm{JSD}\approx\frac{1}{8}\sum_i(\hat{p}_i-\hat{q}_i)^2/p_i$.
By the multivariate CLT, this Pearson-type statistic follows a scaled chi-squared
distribution~\cite{yurdakulNaranjo}. The scaling factor $\tfrac{1}{8}$ carries through
directly. See Appendix~\ref{app:JSD_null} for the full derivation.
\end{proof}

\subsection{Quadratic Approximation of KLD}

\begin{theorem}[KLD under the null]\label{theorem:th2}
Under the conditions of Theorem~\ref{theorem:th1}, if $p_i=q_i$ for all $i$, then
\begin{equation}\label{eq:KLD_null}
  \mathrm{KLD}
  \;\approx\;
  \frac{1}{2}\sum_{i=1}^{B}\frac{(\hat{p}_i-\hat{q}_i)^{2}}{p_i}
  \;\sim\;
  \frac{1}{2}\,\chi^{2}_{B-1}\!\left(\frac{1}{m}+\frac{1}{n}\right).
\end{equation}
\end{theorem}

\begin{proof}
Taylor-expanding $\log\hat{p}_i$ and $\log\hat{q}_i$ to second order, substituting into
$\mathrm{KLD}=\sum_i\hat{p}_i(\log\hat{p}_i-\log\hat{q}_i)$, and setting $p_i=q_i$ collapses
the squared-error terms. Writing $\hat{p}_i=p_i+(\hat{p}_i-p_i)$ in the leading coefficient
then separates the dominant $(\hat{p}_i-\hat{q}_i)^2/(2p_i)$ term from remainders that
vanish as $n,m\to\infty$. The chi-squared result follows by the same Pearson argument
as Theorem~\ref{theorem:th1b}, with scaling factor $\tfrac{1}{2}$.
Full details are in Appendix~\ref{app:KLD}.
\end{proof}

\medskip
\noindent\textbf{Remark.}
Comparing~\eqref{eq:JSD_null} and~\eqref{eq:KLD_null} reveals that
$\mathrm{JSD}\approx\tfrac{1}{4}\,\mathrm{KLD}$ under the null hypothesis $p_i=q_i$,
a relationship that has practical implications for threshold selection when either
measure is used for population stability monitoring.

\subsection{Conditions for Validity of the Chi-Square Approximation}\label{sec:conditions}
The chi-square approximations derived in Theorems~1--3 rest on the asymptotic normality of multinomial proportions and the accuracy of the second-order Taylor expansion. In practice, these approximations are reliable only when certain regularity conditions are satisfied. This subsection makes these conditions explicit and provides guidelines for bin selection, particularly in credit risk applications.
 
\subsubsection{Minimum Expected Bin Count}\label{sec:bin_count}
The chi-square approximation requires that the normal approximation to the multinomial distribution be adequate in every bin. A necessary condition is that each bin contain a sufficient number of expected observations. Formally, for the base sample of size $n$ and target sample of size $m$, the approximation is reliable when
\begin{equation}\label{eq:min_expected}
    n \cdot p_i \geq 5 \quad \text{and} \quad m \cdot q_i \geq 5 \quad \text{for all } i = 1, \ldots, B.
\end{equation}
This is the classical rule for Pearson-type chi-square statistics. When this condition is violated---for example, when a bin has very few observations or when a probability $p_i$ is close to zero---the normal approximation to $\hat{p}_i$ breaks down, the Taylor expansion residuals are no longer negligible, and the chi-square reference distribution becomes unreliable.
 
In addition, the KLD and PSI measures require $\hat{q}_i > 0$ for all bins (due to the logarithmic terms), making them undefined when any target bin is empty. JSD is more robust in this regard because its mixture formulation $\hat{Q}_i = (\hat{p}_i + \hat{q}_i)/2$ ensures $\hat{Q}_i > 0$ whenever at least one of the two samples populates the bin. This is a practical advantage of JSD in applications where sparse bins are common.
 
A more conservative guideline, appropriate when distributions are skewed or when bins span the tails, is to require $n \cdot p_i \geq 10$ and $m \cdot q_i \geq 10$. When these thresholds cannot be met, practitioners should either merge adjacent bins to increase expected counts or use exact (permutation-based) tests instead of the chi-square critical values.
 
\subsubsection{Sample Size Requirements}\label{sec:sample_size}
Beyond the per-bin condition, the overall sample sizes $n$ and $m$ must be large enough for the quadratic (second-order) Taylor expansion to dominate the higher-order remainder terms. As shown in the derivations, the remainder for JSD is $\mathcal{O}(1/m) + \mathcal{O}(1/n)$ and for KLD is $\mathcal{O}(1/\sqrt{m}) + \mathcal{O}(1/\sqrt{n})$, indicating that KLD's approximation converges more slowly and requires larger samples for the same accuracy.
 
Our simulation results in Section~3 provide empirical guidance. The Type~I error analysis (Tables~3--5) shows that:
\begin{itemize}
    \item For $n, m \geq 400$, all three measures maintain rejection rates within $\pm 1\%$ of the nominal 5\% level, indicating that the chi-square approximation is dependable.
    \item For $n = m = 200$ (the smallest combination tested), PSI and KLD exhibit mild inflation of Type~I error (up to 6.1\% and 7.1\%, respectively), while JSD remains well-calibrated at approximately 4.3\%.
    \item Asymmetric sample sizes (e.g., $n = 200$, $m = 1000$) can inflate Type~I error for PSI and KLD even when the larger sample alone would be sufficient, because the approximation quality is governed by $\min(n, m)$.
\end{itemize}
 
As a practical rule of thumb, we recommend $\min(n, m) \geq 400$ when using the chi-square critical values derived in this paper. When $\min(n, m) < 400$, the JSD-based test is preferred over PSI or KLD due to its more conservative Type~I error behavior.
 
\subsubsection{Guidelines for Bin Selection in Credit Risk Applications}\label{sec:bin_selection}
The number of bins $B$ directly affects both the degrees of freedom of the chi-square distribution ($B - 1$) and the expected count per bin ($n/B$ under a uniform distribution). Choosing $B$ therefore involves a bias--variance trade-off: too few bins may mask genuine distributional shifts by aggregating heterogeneous regions, while too many bins may produce sparse counts that violate the minimum expected count condition and inflate Type~I error.
 
A practical starting point is to choose $B$ such that the minimum expected count per bin satisfies~\eqref{eq:min_expected}. Under a roughly uniform distribution, this requires $B \leq n/5$ (or $B \leq n/10$ under the more conservative threshold). For example, with $n = m = 1{,}000$, this suggests $B \leq 200$ at most, though values in the range $B \in [10, 20]$ are typical in practice and consistent with our simulation setup of $B = 10$.

Table~\ref{tab:conditions} summarizes the conditions under which the chi-square approximations can be safely applied.
 
\begin{table}[H]
\centering
\caption{Summary of Conditions for Valid Application of Chi-Square Approximations}
\begin{tabular}{p{4.2cm}p{7cm}}
\hline
\textbf{Condition} & \textbf{Requirement} \\ \hline
Minimum expected count & $n \cdot p_i \geq 5$ and $m \cdot q_i \geq 5$ for all bins (conservative: $\geq 10$) \\
No empty bins & Required for PSI and KLD; JSD tolerates one-sided empty bins via the mixture $\hat{Q}_i$ \\
Minimum sample size & $\min(n, m) \geq 400$ recommended; for $\min(n, m) < 400$, prefer JSD \\
Number of bins & $B \leq \min(n, m) / 5$; typically $B \in [10, 20]$ for credit risk \\
Bin construction & Use rating-grade--aligned or equal-frequency bins; avoid equal-width bins for skewed PD distributions \\
Robustness check & Report results for multiple values of $B$ \\ \hline
\end{tabular}
\label{tab:conditions}
\end{table}
 
When any of these conditions cannot be satisfied for instance, when monitoring a low-default portfolio where many rating grades contain fewer than five defaults, practitioners should supplement or replace the chi-square--based test with permutation tests or bootstrap-based confidence intervals for the divergence statistics, which do not rely on the asymptotic approximation.

\section{Numerical Experiments and Simulations}
A simulation is, in this section, run to demonstrate the convergence of the respective distributions of divergence measures to better understand how these measures behave based on different distributions. The consistency of these divergence measures is then compared and assessed by conducting statistical power analysis.
 
\subsection{Asymptotic Distribution Verifications by Simulations}
We will now conduct numerical experiments to test the validity of the approximate asymptotic distributions that were derived in the previous section. To do this, we generated 10,000 simulations of a pair of multinomial distributions. At each simulation a divergence measure is obtained. Once we had a large sample of these indicators, we compared the histogram to the proposed distribution, which is chi-square. The simulation uses $B = 10$ bins with $n = m = 500$, satisfying the conditions outlined in Section~\ref{sec:conditions}: each bin has an expected count of 50, well above the minimum threshold of 5.
 
Figure~\ref{fig:DivergenceSimHist} presents the empirical validation of the theoretical asymptotic distributions for three divergence measures. The results demonstrate agreement between empirical distributions from 10,000 Monte Carlo simulations and their corresponding theoretical $\chi^{2}$ distributions. In the right panel of the same figure, a quantile-quantile (Q-Q) plot show exceptional linear relationships with $R^2$ values exceeding 0.99 for all three measures (PSI: 0.9927, KLD: 0.9912, JSD: 0.9932). The empirical quantiles closely align with the theoretical $\chi^{2}(9)$ quantiles, confirming the validity of the asymptotic approximations. 
 
Table~\ref{tab:validation_results} summarizes the validation statistics. The empirical means closely approximate their theoretical counterparts: PSI (0.0183 vs. 0.0180), KLD (0.0091 vs. 0.0090), and JSD (0.0023 vs. 0.0023). Similarly, the empirical standard deviations align well with theoretical predictions. Kolmogorov-Smirnov tests yield p-values of 0.030, 0.076, and 0.116 for PSI, KLD, and JSD respectively, indicating no significant deviation from the $\chi^2$ distributions at conventional significance levels.
 
\begin{figure}[H]%
    \centering
    \includegraphics[width=1\textwidth]{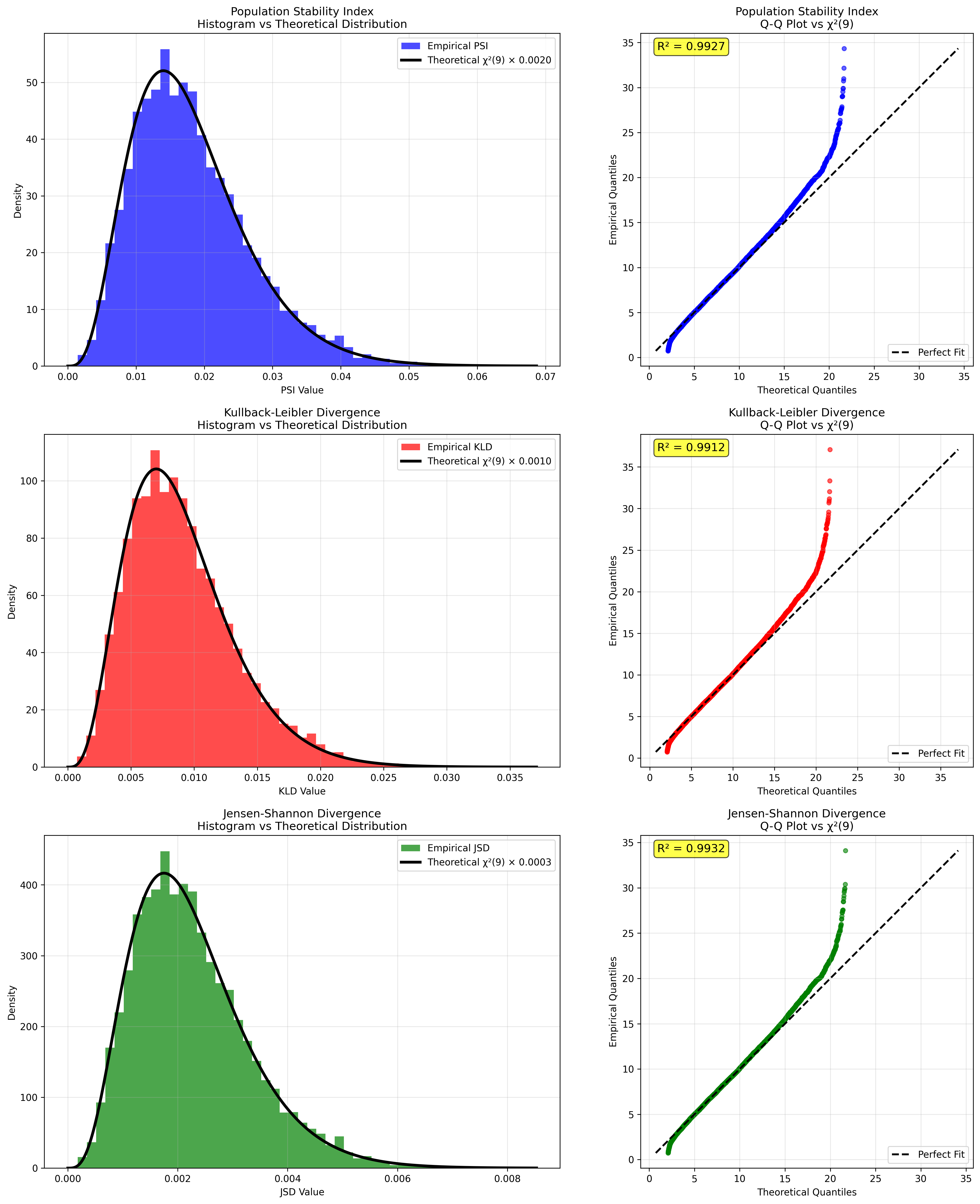}
    \captionsetup{justification=raggedright,singlelinecheck=false}
    \caption{Empirical Distributions of Divergence Indicators Histograms and Asymptotically Exact Densities}%
    \label{fig:DivergenceSimHist}%
\end{figure}
 
\begin{table}[ht]
\centering
 \caption{Asymptotic Distribution Validation Results}
\small
\setlength{\tabcolsep}{5pt}
\begin{tabular}{lcccccc}
\hline
Measure & \begin{tabular}[c]{@{}c@{}}Empirical\\Mean\end{tabular} & \begin{tabular}[c]{@{}c@{}}Theoretical\\Mean\end{tabular} & \begin{tabular}[c]{@{}c@{}}Empirical\\Std\end{tabular} & \begin{tabular}[c]{@{}c@{}}Theoretical\\Std\end{tabular} & \begin{tabular}[c]{@{}c@{}}KS\\Statistic\end{tabular} & \begin{tabular}[c]{@{}c@{}}KS\\p-value\end{tabular} \\ 
\hline
PSI & 0.0183 & 0.0180 & 0.0086 & 0.0085 & 0.0145 & 0.0301 \\
KLD & 0.0091 & 0.0090 & 0.0043 & 0.0042 & 0.0128 & 0.0758 \\
JSD & 0.0023 & 0.0023 & 0.0011 & 0.0011 & 0.0119 & 0.1158 \\
\hline
\end{tabular}
\label{tab:validation_results}
\end{table}
 
Thus, these results indicate that, as is shown in the equation~\eqref{eq:JSD_null} and \eqref{eq:KLD_null}, JSD and KLD follow $\frac{1}{8}\chi^{2}_{B-1}\left(\frac{1}{m}+\frac{1}{n}\right)$ and $\frac{1}{2}\chi^{2}_{B-1}\left(\frac{1}{m}+\frac{1}{n}\right)$, respectively. This is already proved in \cite{yurdakulNaranjo} for PSI, which has approximately $\chi^{2}_{B-1}\left(\frac{1}{m}+\frac{1}{n}\right)$.
 
Figure~\ref{fig:DivergenceSimQQ} displays the convergence of theoretical and empirical cumulative distribution functions derived from simulated data in the top panels. The supremum distances $\sup|F(x)-\hat{F}(x)|\sim 10^{-2}$ for each divergence measure demonstrate excellent convergence, making the theoretical $\chi^{2}$ and empirical CDFs virtually indistinguishable.
 
The bottom panels analyze the quantile ratios of JSD to PSI and KLD to PSI across all probability levels. These ratios provide direct empirical validation of the theoretical scaling relationships. The theoretical ratios emerge directly from these scaling factors:
\begin{align}
\frac{\text{JSD}}{\text{PSI}} &= \frac{\frac{1}{8} \times \chi^2_{B-1} \times \left(\frac{1}{m} + \frac{1}{n}\right)}{\chi^2_{B-1} \times \left(\frac{1}{m} + \frac{1}{n}\right)} = \frac{1}{8} \label{eq:jsd_psi_ratio} \\
\frac{\text{KLD}}{\text{PSI}} &= \frac{\frac{1}{2} \times \chi^2_{B-1} \times \left(\frac{1}{m} + \frac{1}{n}\right)}{\chi^2_{B-1} \times \left(\frac{1}{m} + \frac{1}{n}\right)} = \frac{1}{2} \label{eq:kld_psi_ratio}
\end{align}
 
\begin{figure}[H]%
   \centering
   \includegraphics[width=1\textwidth]{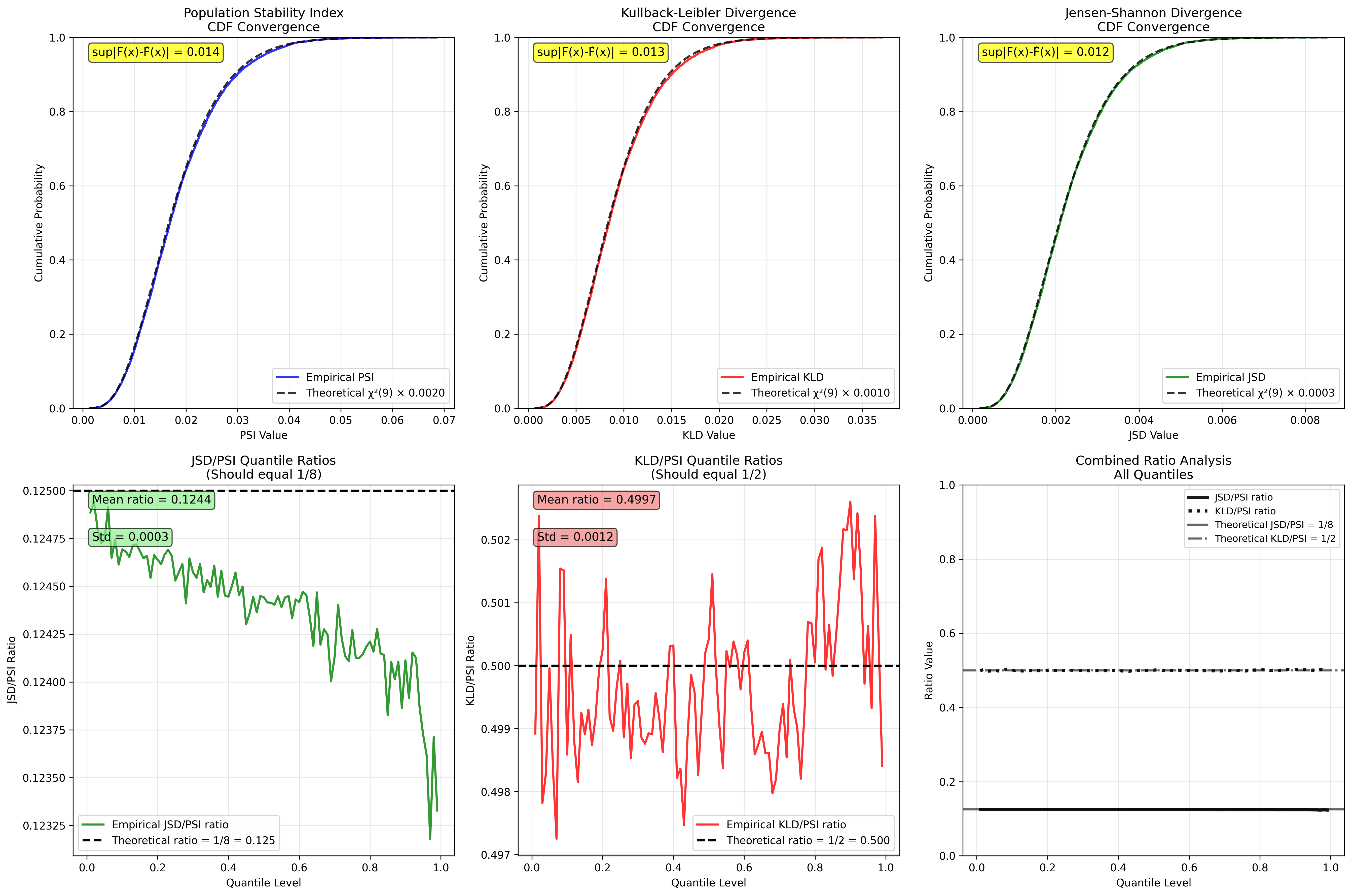}
   \captionsetup{justification=raggedright,singlelinecheck=false}
   \caption{Divergence Indicator Empirical vs Exact CDF and Quantile Ratios}%
   \label{fig:DivergenceSimQQ}%
\end{figure}
 
Figure~\ref{fig:DivergenceSimQQ} confirms these theoretical discussion empirically: the JSD/PSI ratio consistently approaches $\frac{1}{8}$ and the KLD/PSI ratio approaches $\frac{1}{2}$ across all quantiles, with minimal variability around the theoretical values. This quantile-level validation provides strong evidence that the asymptotic scaling relationships hold uniformly across the entire distribution, not just for central tendencies.
 
\subsection{Statistical Power Analysis}
This part includes simulation analyses to compare three divergence measures and investigate their statistical power. Conducting a rigorous statistical power analysis across diverse sample sizes and statistical moments is helpful for methodological precision and reliability in the assessment of distributional stability. The analysis aims to show the sensitivity of these divergence measures to varying degrees of distributional changes, offering insights into their robustness and effectiveness under different conditions. Evaluating statistical power across multiple moments allows for a nuanced understanding of the measures' performance in capturing diverse statistical characteristics.
 
To do this, we obtain $\chi^2$ benchmark values for each divergence measure at different sample sizes and moments. In that respect, the statistical power analysis is conducted to evaluate the efficacy of divergence measures in discerning distributional changes between two distributions over multiple trials.
 
Now, two normally distributed data are obtained and divergence measures are applied on the densities obtained through the histogram of these simulated distributions. These simulated data are manipulated via different mean and standard deviation parameters to systematically examine the sensitivity of divergence measures under diverse scenarios. In order to simulate our data 1,000 normal variates are generated with given $\mu$ and $\sigma$ parameters.
 
$\mathcal{N}\left(\mu_{1}=0,\sigma_{1}=1\right)$ and $\mathcal{N}\left(\mu_{2}=0.25,\sigma_{2}=1\right)$ then generated their histograms for bins of $\left(200, 400, 600, 800, 1000\right)$. 
 
Given this setup, we calculate test statistics for PSI, JSD, and KLD divergence measures and compare them with critical values coming from corresponding $\chi^2(k)$ distribution obtained in Table~\ref{table:diff_psi}, \ref{table:diff_kld}, and \ref{table:diff_jsd} below. The test hypothesis is given as follows: 
\begin{align}
    H_{0}&:\hat{p}^{1}_{i}=\hat{p}^{2}_{i} \nonumber\\
    H_{1}&:\hat{p}^{1}_{i}\neq\hat{p}^{2}_{i} \nonumber
\end{align}
where $\hat{p}^{1}$ is the predetermined distribution and $\hat{p}^{2}$ denotes the empirical distribution that we obtain from the data.
 
\subsubsection{Scenario A: Identical Distributions (Type I Error Control)}
When both distributions are identical ($\mu=0$, $\sigma=1$ for both), an ideal measure should reject $H_0$ approximately 5\% of the time at the $\alpha=0.05$ significance level.
 
\begin{table}[H]
\centering
\caption{Power Analysis of PSI Measure with Identical Distributions ($\mu=0$, $\sigma=1$)}
\begin{tabular}{cccccc}
\hline
n/m &    200  &    400  &    600  &    800  &    1000 \\ \hline
200 &  0.0558 &  0.0521 &  0.0587 &  0.0591 &  0.0643 \\
400 &  0.0489 &  0.0441 &  0.0419 &  0.0537 &  0.0455 \\
600 &  0.0415 &  0.0388 &  0.0425 &  0.0441 &  0.0473 \\
800 &  0.0388 &  0.0420 &  0.0386 &  0.0437 &  0.0419 \\
1000 & 0.0391 &  0.0397 &  0.0413 &  0.0416 &  0.0407 \\ \hline
\end{tabular}
\label{tab:PSI_null}
\end{table}
 
\begin{table}[H]
\centering
\caption{Power Analysis of KLD Measure with Identical Distributions ($\mu=0$, $\sigma=1$)}
\begin{tabular}{cccccc}
\hline
n/m &    200  &    400  &    600  &    800  &    1000 \\ \hline
200 &  0.0609 &  0.0572 &  0.0608 &  0.0647 &  0.0708 \\
400 &  0.0433 &  0.0404 &  0.0384 &  0.0491 &  0.0418 \\
600 &  0.0379 &  0.0355 &  0.0389 &  0.0404 &  0.0433 \\
800 &  0.0355 &  0.0384 &  0.0353 &  0.0400 &  0.0383 \\
1000 & 0.0357 &  0.0363 &  0.0378 &  0.0380 &  0.0372 \\ \hline
\end{tabular}
\label{tab:KLD_null}
\end{table}
 
\begin{table}[H]
\centering
\caption{Power Analysis of JSD Measure with Identical Distributions ($\mu=0$, $\sigma=1$)}
\begin{tabular}{cccccc}
\hline
n/m &    200  &    400  &    600  &    800  &    1000 \\ \hline
200 &  0.0434 &  0.0436 &  0.0421 &  0.0429 &  0.0430 \\
400 &  0.0448 &  0.0336 &  0.0386 &  0.0388 &  0.0362 \\
600 &  0.0445 &  0.0391 &  0.0357 &  0.0358 &  0.0365 \\
800 &  0.0456 &  0.0381 &  0.0342 &  0.0361 &  0.0350 \\
1000 & 0.0466 &  0.0391 &  0.0338 &  0.0353 &  0.0376 \\ \hline
\end{tabular}
\label{tab:JSD_null}
\end{table}
 
All three measures demonstrate appropriate Type I error control, with rejection rates moving around the nominal 5\% level across all sample size combinations. JSD exhibits the most conservative behavior, with rejection rates consistently at or slightly below 5\%, indicating it is least prone to false positives. KLD shows slightly higher variability in smaller samples (ranging from 3.5\% to 7.1\% for n,m $\leq$ 400) compared to PSI and JSD, though both stabilize as sample size increases. These results confirm that all measures maintain proper statistical calibration when distributions are truly identical, and empirically validate the sample size recommendations in Section~\ref{sec:sample_size}: for $\min(n, m) \geq 400$, all three measures stay within $\pm 1\%$ of the nominal level, while the $n = m = 200$ case shows that KLD's slower convergence rate ($\mathcal{O}(1/\sqrt{n})$ remainder versus $\mathcal{O}(1/n)$ for JSD) manifests as higher Type I error inflation in small samples.
 
\subsubsection{Scenario B: Shifted Distributions (Statistical Power)}
When distributions genuinely shifted ($\mu_1=0$ versus $\mu_2=0.25$, both with $\sigma=1$), higher rejection rates indicate greater statistical power, that is to say, the ability to correctly detect the distributional shift.
 
\begin{table}[H]
\centering
\caption{Power Analysis of PSI Measure with Shifted Distributions ($\mu_2=0.25$, $\sigma=1$)}
\begin{tabular}{cccccc}
\hline
n/m &    200  &    400  &    600  &    800  &    1000 \\ \hline
200 &  0.3229 & 0.4044 & 0.4744 & 0.4435 & 0.4953 \\
400 &  0.3951 & 0.5653 & 0.7329 & 0.6652 & 0.7605 \\
600 &  0.4405 & 0.6731 & 0.8472 & 0.7898 & 0.8822 \\
800 &  0.4619 & 0.7221 & 0.9113 & 0.8428 & 0.9401 \\
1000 & 0.4734 & 0.7720 & 0.9404 & 0.8795 & 0.9649 \\ \hline
\end{tabular}
\label{tab:PSI_power}
\end{table}
 
\begin{table}[H]
\centering
\caption{Power Analysis of KLD Measure with Shifted Distributions ($\mu_2=0.25$, $\sigma=1$)}
\begin{tabular}{cccccc}
\hline
n/m &    200  &    400  &    600  &    800  &    1000 \\ \hline
200 &  0.3182 & 0.4235 & 0.4935 & 0.4669 & 0.5110 \\
400 &  0.3837 & 0.5744 & 0.7315 & 0.6651 & 0.7787 \\
600 &  0.4213 & 0.6624 & 0.8525 & 0.7866 & 0.8875 \\
800 &  0.4504 & 0.7232 & 0.9069 & 0.8447 & 0.9357 \\
1000 & 0.4729 & 0.7633 & 0.9398 & 0.8823 & 0.9651 \\ \hline
\end{tabular}
\label{tab:KLD_power}
\end{table}
 
\begin{table}[H]
\centering
\caption{Power Analysis of JSD Measure with Shifted Distributions ($\mu_2=0.25$, $\sigma=1$)}
\begin{tabular}{cccccc}
\hline
n/m &    200  &    400  &    600  &    800  &    1000 \\ \hline
200 &  0.2725 & 0.3759 & 0.4533 & 0.4216 & 0.4752 \\
400 &  0.3694 & 0.5541 & 0.7227 & 0.6519 & 0.7520 \\
600 &  0.4014 & 0.6624 & 0.8476 & 0.7751 & 0.8783 \\
800 &  0.4428 & 0.7172 & 0.9015 & 0.8390 & 0.9380 \\
1000 & 0.4572 & 0.7577 & 0.9414 & 0.8748 & 0.9658 \\ \hline
\end{tabular}
\label{tab:JSD_power}
\end{table}
 
The power analysis reveals several important patterns:
\begin{itemize}
    \item Sample Size Effect: All measures show dramatic power increases with sample size. For example, PSI's power increases from 32\% (n=m=200) to 96\% (n=m=1000), demonstrating that larger samples are crucial for reliably detecting distributional shifts.
    
    \item Balanced Samples: Power is highest when both sample sizes are large and approximately equal. Asymmetric samples (e.g., n=200, m=1000) yield lower power than balanced large samples (n=m=1000).
\end{itemize}
 
PSI demonstrates the highest power, particularly in large samples (96.5\% at n=m=1000), while KLD shows nearly identical performance to PSI in large samples. JSD shows lower power in small-to-medium samples but converges to nearly identical power as PSI and KLD in large samples. Thus, JSD is more conservative in the sense that it requires larger samples to achieve the same power as PSI and KLD. However, it is important to note that this conservatism under the null hypothesis may partly reflect slower convergence of the empirical JSD to its asymptotic chi-square distribution, rather than an inherently superior calibration property. In small samples, the chi-square critical values may not fully capture the true finite-sample distribution of JSD, leading to under-rejection. Developing small-sample corrections or bootstrap-based calibration of the critical values where the null distribution is estimated via resampling rather than the asymptotic formula is a promising direction for future work, particularly for low-default portfolios where sample sizes are limited.
 
This conservative behavior makes JSD particularly valuable when the cost of false positives (incorrectly flagging distributions as different) is high. In model monitoring applications, for instance, JSD reduces unnecessary alerts while still detecting substantive distributional shifts when samples are sufficiently large.

\section{Application of Divergence to Credit Risk}
\subsection{Structural Credit Risk Models}
The integration of divergence measures emerges as a strong method to refine the precision and adaptability of credit risk models. In this part, using the structural credit risk models belonging to well-known stochastic asset processes, we aim to explore the application of divergence measures in detecting distributional changes for credit risk models. The specific focus is on the Merton model and its extensions, including Merton with Jump and Stochastic Volatility with Jump models. The reason behind using these processes is that there is an economic intuition behind them and it is possible to produce meaningful default probabilities without resorting to a pile of dataset. We note, however, that our empirical demonstration relies entirely on simulated asset paths from these structural models rather than actual market or default data. Validation on historical default databases (e.g., Moody's or S\&P default histories) or real-world portfolio monitoring data is needed to confirm the operational value of these divergence measures in production monitoring systems.
 
The motivation for incorporating divergence measures into credit risk modeling lies in their inherent capacity to capture distributional changes in the evolving dynamics of credit risk factors. The Merton model, a pioneering structural credit risk model, traditionally relies on the assumption of normally distributed asset values. However, financial markets often exhibit characteristics that deviate from this simplistic assumption, necessitating a more comprehensive modeling approach. The Merton with Jump model introduces jumps, accommodating abrupt changes in firm values. Meanwhile, the Stochastic Volatility with Jump model incorporates additional complexity through the inclusion of stochastic volatility, capturing the inherent variability in financial markets. Applying divergence measures on these models serve as diagnostic tools, enabling a systematic assessment of distributional shifts in the efficacy of credit risk models. Through this analysis, we will be able to uncover insights into the applicability and advantages of divergence measures in dynamic financial world.
 
Let us start by briefly explaining the structural credit risk models through which we obtain default probabilities. Our first task in deriving the stochastic volatility with the jump model using the Bates model \cite{tankov}.
\begin{equation}
A(T)=A(t)\exp{\left(\left(r-\frac{v(\tau)}{2} -\lambda \mu\right)\tau+\sigma \sqrt{v(\tau)}W(\tau)+\sum_{k=0}^{N(\tau)}\log(1+Y_{i})\right)}\nonumber
\end{equation}
where $\log(1+Y)\sim\mathcal{N}\left(\log(1+\mu_{J})-\frac{\sigma_{J}^{2}}{2},\sigma_{J}^{2}\right)$
\begin{align}
\label{eq:BatesChf}
C(\tau)={\tau \kappa \theta}
\frac{\left(\kappa-iu\rho_{x\nu}\sigma_v\right)}{\sigma_{v}^2}-\frac{2\kappa \theta}{\sigma_{v}^2} \ln\left( \frac{(\kappa-iu\rho_{x\nu}\sigma_v)}{\sqrt{\Delta}}\sinh\left(\frac{\sqrt{\Delta}\tau}{2}\right)+\cosh\left(\frac{\sqrt{\Delta}\tau}{2}\right)\right)\nonumber\\
\end{align}
where
\begin{equation}
\sqrt{\Delta}=\sqrt{\left(iu\rho_{x\nu}\theta -\kappa \right)^2+\sigma_{v}^2 \left(u^2+iu \right)}\nonumber
\end{equation}
\begin{equation}
B(\tau)=\frac{-(u^2+iu)}{\left(\kappa-iu\rho_{x\nu}\theta\right)+\sqrt{\Delta}\coth\left(\frac{\sqrt{\Delta}\tau}{2}\right)}\nonumber
\end{equation}
\begin{equation}
\varphi_{sv}\left(u,\tau,v(t),\sigma_{J},\mu_{J},\kappa,\rho,\theta \right)=e^{\left(C(\tau)+B(\tau)v(t)+iu\log A(t)+\lambda\tau\left[(1+\mu)^{ui}\exp{\left(\frac{-\sigma_{J}^{2}(u^2+ui)}{2}\right)}-1\right]\right)}\nonumber
\end{equation}
 
For a given liability level $L$, we can concisely derive stochastic volatility jump model PD as follows,
\begin{equation}
    P\left(A(T)\leq L\right)=\frac{1}{2}-\frac{1}{\pi}\int_{0}^{\infty}Re\left( \frac{e^{-iu\ln(L)} \varphi_{sv}\left(u,A(t),\tau,v(t),\sigma_{J},\mu_{J},\kappa,\rho,\theta \right)}{iu}du
\right)
\label{eq:BatesPD}
\end{equation}
The second model is Merton model with jump. Let us start with the formula of asset pricing process \cite{tankov}:
\begin{equation}
A(T)=A(t)\exp{\left(\left(r-\frac{\sigma^{2}}{2}-\lambda \mu\right)\tau+\sigma W(\tau)+\sum_{k=0}^{N(\tau)}y_{i}\right)}\nonumber.    
\end{equation}
where $y\sim \mathcal{N}\left(\mu_{j},\sigma_{j}\right)$.
Then the characteristic function,
\begin{equation}
   \varphi(u,\tau,A(t),\mu_{J},\sigma,\sigma_{J})= e^{iu\log\left(A(t)\right)+\tau\left[iu\left(r-\frac{\sigma^{2}}{2}\left(1-iu\right)-\lambda\mu_{J} \right)+\lambda \left(\exp\left(iu\mu_{J}-\frac{u^{2}\sigma_{J}^{2}}{2}\right) -1\right)\right]}.
   \label{eq:MertonJumpChf}
\end{equation}
Likewise, using characteristic function it is convenient to obtain the probability of default equation as follows:
\begin{equation}
    P\left(A(T)\leq L\right)=\frac{1}{2}-\frac{1}{\pi}\int_{0}^{\infty}Re\left( \frac{e^{-iu\log(L)} \varphi(u,\tau,A(t),\mu_{J},\sigma,\sigma_{J})}{iu}du
\right)
\label{eq:mertonJumpPD}
\end{equation}
 
Final model is baseline Merton model, using Geometric Brownian Motion, probability of default equation can be directly obtained from Black-Scholes and Merton model $\Phi(-d_{-})$ and it takes the following form:
\begin{equation}
P\left(A(T)\leq L\right)=\Phi\left(\frac{\log\left(\frac{L}{A(t)}\right)-(r-\frac{\sigma^{2}}{2})\tau}{\sigma\sqrt{\tau}} \right)
\label{eq:mertonPD}
\end{equation}
where in equations~\eqref{eq:BatesChf}, \eqref{eq:BatesPD}, \eqref{eq:MertonJumpChf}, \eqref{eq:mertonJumpPD} and \eqref{eq:mertonPD}, the parameters $\sigma, \lambda, \mu_{J}, \sigma_{J}$, L, and r correspond to volatility of asset, jump frequency, jump mean, jump volatility, risk-neutral interest rate and liability level, where applicable. Additionally, in equation~\eqref{eq:BatesChf} and \eqref{eq:BatesPD} the parameters $\sigma_{v}, \rho_{xv}, \theta, \text{and } \kappa$ correspond to volatility of the asset volatility, correlation between asset volatility and asset price, long-run average volatility and asset volatility adjustment parameter.
 
\subsection{Numerical Experiments for Credit Risk}
Having discussed the distributions of divergence measures and structural credit risk models, it is worthwhile to discuss the results of the credit risk models that are being tested with these measures. As expected, Figure~\ref{fig:PDmodels} shows that the probability of default estimated by the Merton model with jump and SVJ is higher than that of the regular Merton model. However, over time, the models exhibit convergent behavior in their trajectories, even though the Merton model with jump initially overshoots SVJ. Despite this visual convergence in time series patterns, statistical divergence tests reveal that the distributional characteristics of the estimated default probabilities remain significantly different across models.
 
\begin{figure}[H]
\centering
\includegraphics[width=\linewidth]{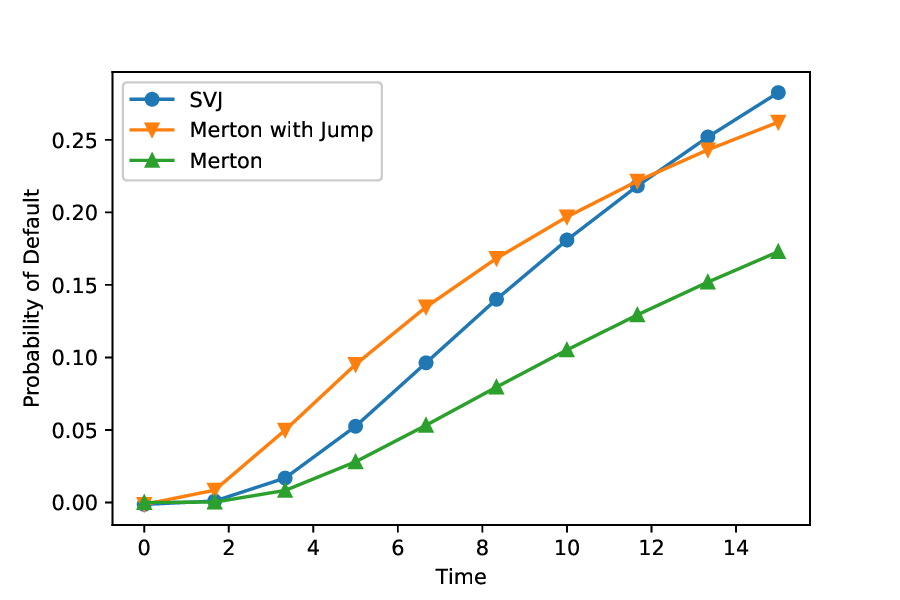}
\captionsetup{justification=raggedright,singlelinecheck=false}
\caption{Default Models with respect to Time}
\label{fig:PDmodels}
\end{figure}
 
The divergence measures mentioned in this study will now be applied to the results of the credit risk models, and the results will be assessed. Before delving into the statistical analysis, Figure~\ref{fig:PSIgraph} provides valuable information. Figure~\ref{fig:PSIgraph} shows that all the credit risk models diverge from each other to some degree. The first panel (left) compares the Merton model with jump and the SVJ model, showing slight divergence. The second panel (middle) compares the Merton model and the SVJ model, also exhibiting slight divergence similar to the first panel. However, the third panel (right), which compares the Merton model with jump and the standard Merton model, clearly exhibits divergence.
 
\begin{figure}[H]
    \centering
    \includegraphics[width=\linewidth]{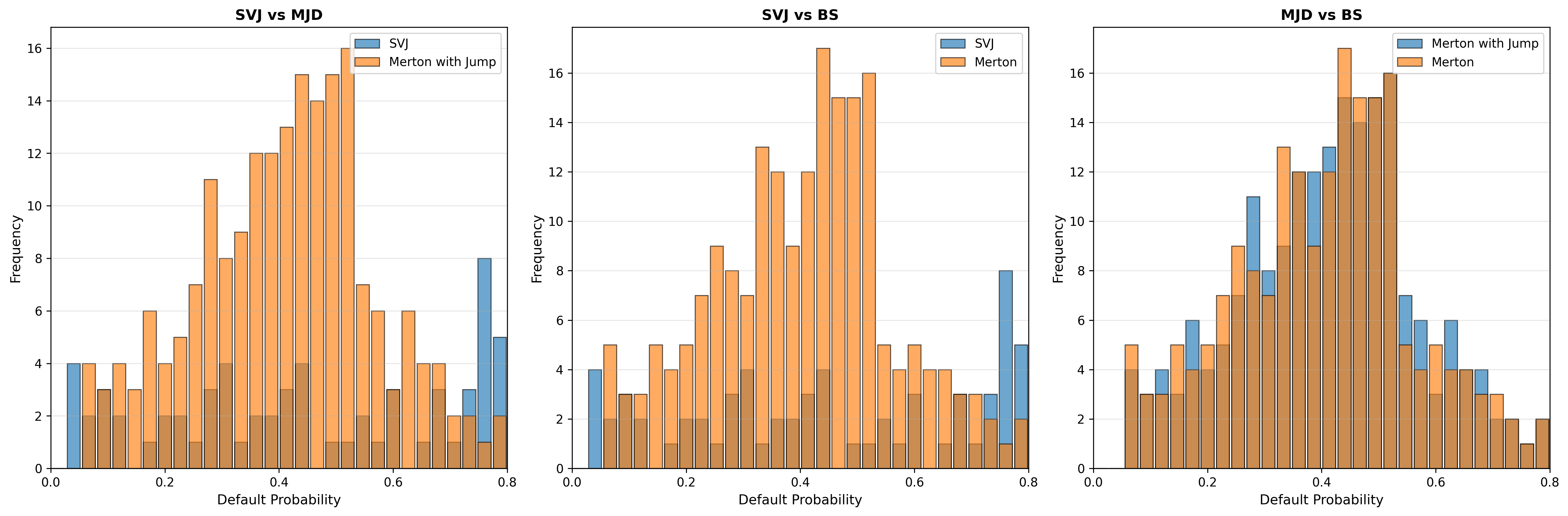}
    \captionsetup{justification=raggedright,singlelinecheck=false}
    \caption{Probability of Default Densities}
    \label{fig:PSIgraph}
\end{figure}
 
Visual inspection of Figure~\ref{fig:PSIgraph} suggests distributional differences, particularly evident in the lower panels comparing Merton-SVJ and Merton with Jump-Merton. However, visual assessment alone cannot quantify the statistical significance or magnitude of these differences. To rigorously evaluate whether observed divergences exceed what might occur by random variation, we apply formal hypothesis testing using the three divergence measures.
 
Thus, critical values for all these divergence measures are derived at 5\% significance level and reported at the bottom line of Table~\ref{table:diff_psi}, \ref{table:diff_kld}, and \ref{table:diff_jsd}. These values are 0.0991, 0.0495, and 0.0123 for PSI, KLD, and JSD, respectively. In this application, we use $B = 10$ equal-frequency bins constructed from the base (Merton) distribution. This choice satisfies the validity conditions of Section~\ref{sec:conditions}: with the simulated portfolio of 1,000 obligors, each bin contains approximately 100 observations, well above the minimum expected count of 5. Following the recommendation in Section~\ref{sec:bin_selection}, practitioners monitoring PD distributions should consider using rating-grade--aligned bins when a master scale is available, or equal-frequency bins as employed here when the goal is to detect shifts across the entire PD range.
 
As we showed in Section~\ref{sec:2}, divergence measures considered follow weighted or scaled $\chi^2$ distributions. Thus, using $\chi^2_{B-1}\left(\frac{1}{m}+\frac{1}{n} \right)$ critical value of 0.0991 for PSI, it is possible to see if a distribution is the same with the predetermined distribution. In other words, the test is applied to see whether there is any difference between distributions. 
 
The results of these analyses are shown in Table~\ref{table:diff_psi}, \ref{table:diff_kld}, and \ref{table:diff_jsd}. According to the results, the largest difference in default probabilities is observed in the PSI measure. Additionally, Table~\ref{table:diff_psi} reveals differences in the probability of default distribution among the models. The $\chi^2$ benchmark for the PSI measure is 0.0991 and the differences in probability of default are higher than that, suggesting that the distribution of these results is not the same. It is important to note that the $\chi^2$ critical value is obtained at a 5\% significance level for all of the analyses. This finding is also supported by Figure~\ref{fig:PSIgraph}, which shows a divergence between the last two plots.

\begin{table}[h!]
\caption{Convergence of PD Models Based on PSI}
\centering
\label{table:diff_psi}
\begin{tabular}{ccc}
\hline
PD Models & Difference & Statistics \\ \hline
Merton with Jump - SVJ & 0.2809 & \\
Merton - SVJ & 0.7420&\\
Merton with Jump - Merton & 1.9419&\\
$\chi^2_{B-1}\left(\frac{1}{m}+\frac{1}{n} \right)$ critical value & & 0.0991\\ \hline
\end{tabular}%
\end{table}
 
\begin{table}[H]
\caption{Convergence of PD Models Based on KLD}
\centering
\label{table:diff_kld}
\begin{tabular}{ccc}\hline
PD Models & Difference & Statistics \\ \hline
Merton with Jump - SVJ & 0.1465 & \\ 
Merton - SVJ & 0.3497&\\
Merton with Jump - Merton & 0.9843&\\
$\frac{1}{2}\chi^2_{B-1}\left(\frac{1}{m}+\frac{1}{n} \right)$ critical value & & 0.0495\\ \hline
\end{tabular}%
\end{table}
 
\begin{table}[H]
\caption{Convergence of PD Models Based on JSD}
\centering
\begin{tabular}{ccc}
\hline
 PD Models & Difference & Statistics \\ \hline
Merton with Jump - SVJ & 0.0310& \\
Merton - SVJ & 0.0841&\\
Merton with Jump - Merton & 0.1944&\\
$\frac{1}{8}\chi^2_{B-1}\left(\frac{1}{m}+\frac{1}{n} \right)$ critical value & & 0.0123\\ \hline
\end{tabular}%
\label{table:diff_jsd}
\end{table}
 
These different magnitudes across measures, PSI ranging from 0.28 to 1.94, KLD from 0.15 to 0.98, and JSD from 0.03 to 0.19, reflect the inherent scaling differences in how each measure quantifies divergence. PSI uses logarithmic differences without symmetrization, KLD measures information gain from one distribution to another, while JSD symmetrizes and bounds the divergence. Critically, what matters is not the absolute magnitude but whether each value exceeds its respective critical threshold, which all measures consistently achieve in our analysis.
 
The result obtained from KLD measure is similar to that of PSI and JSD measures. Numerically, difference in probability of default estimation between Merton model with Jump and SVJ is 0.1465 whereas difference is 0.2809 in PSI approach. The difference is even larger in the case of Merton - SVJ models (0.3497) and Merton with Jump - Merton (0.9843) models. Using KLD measure, we ended up with rejecting the null hypothesis and concluded that the probability of default distributions are different.
 
Now, it is worth discussing the result obtained from JSD measure. Similarly, our test based on JSD measure shows that the default probability distributions deviate from each other, as all divergence values exceed the $\chi^2$ benchmark value of 0.0123. Notably, JSD produces smaller absolute divergence values (ranging from 0.0310 to 0.1944) compared to PSI (0.2809 to 1.9419) and KLD (0.1465 to 0.9843). However, this does not indicate different conclusions—all three measures unanimously reject the null hypothesis of distributional equality. The smaller magnitude reflects JSD's mathematical scaling properties rather than conservatism in decision-making.
 
The key insight is that while the absolute divergence values differ across measures due to their distinct mathematical formulations, all three measures provide consistent statistical conclusions: the credit risk models produce significantly different probability of default distributions. This consistency validates our earlier power analysis findings that, when distributional differences are substantial, all measures converge to similar detection capabilities.
 
\section{Conclusion}
In today's volatile world, it is needless to say that distributions of the data used in a model should be closely monitored. Otherwise, the model becomes useless or misleading over time. Given the complexity and intrication of the financial world, monitoring is even more important. This study is an attempt to provide a solution to this issue by deriving statistical properties of JSD and KLD measures and empirical findings obtained from PSI, KLD, and JSD.
 
Based on $\chi^2$ benchmark values, it is endeavoured to detect distributional difference of probability of default derived from credit risk models. The result suggests that JSD is more conservative in detecting distributional difference as the $\chi^2$ benchmark value is the lowest. However, KLD and PSI tend to detect distributional difference even if the difference is not much pronounced. Besides, statistical power analysis is employed the consistency of the analysis over different sample sizes. Given different mean and standard deviation values, this analysis also confirms that the JSD is more inclined to reject dissimilarity. However, the differences across the divergence measures disappears with large sample size.
 
Based on benchmark values for $\chi^2$, we endeavor to detect differences in the distribution of probability of default from credit risk models. The results suggest that JSD is more conservative in detecting these differences, as its critical value is the lowest. However, KLD and PSI tend to detect distributional difference even if the difference is not much pronounced. Additionally, a statistical power analysis is used to ensure consistency of the analysis across different sample sizes. This analysis also confirms that JSD is more likely to reject dissimilarity, due to variations in mean and standard deviation values. However, with larger sample sizes, the differences among the divergence measures disappear.
 
The divergence indicators that are studied here could be a tool to quantify model risk. We have also established explicit conditions, minimum expected bin counts, sample size requirements, and bin construction guidelines, under which the chi-square approximations can be safely applied in practice, particularly for credit risk model monitoring where PD distributions are typically skewed and low-default portfolios may produce sparse bins. Therefore, further research should be conducted to investigate the distribution implied by other divergence indicators. Additionally, the approximate distributions derived from these measures could be used to quantify model risk and be utilized in pricing or as an add-on for capital calculations. Two further directions merit attention: first, the development of small-sample corrections or bootstrap-based calibration of the chi-square critical values, which would improve the reliability of all three measures and JSD in particular when sample sizes are limited; second, validation of the proposed framework on historical default databases and real-world portfolio monitoring data, which would confirm the operational value of these divergence measures beyond the simulated setting studied here.

\bibliographystyle{plainnat}
\bibliography{bmc_article}

\begin{appendix}
\section{Appendix}
\subsection{Detailed derivation of Theorem~\ref{theorem:th1}}\label{app:JSD}

We provide the full algebraic details omitted from the proof sketch. The Taylor expansions are:
\begin{equation*}
\begin{aligned}
\log(\hat{p}_{i})&=\log(p_{i})+\frac{(\hat{p}_{i}-p_{i})}{p_{i}}-\frac{1}{2}\frac{(\hat{p}_{i}-p_{i})^2}{p_{i}^2}+\frac{1}{3}\frac{(\hat{p}_{i}-p_{i})^3}{p_{i}^3}+\mathcal{O}\!\left(n^{-4}\right),\\
\log(\hat{q}_{i})&=\log(q_{i})+\frac{(\hat{q}_{i}-q_{i})}{q_{i}}-\frac{1}{2}\frac{(\hat{q}_{i}-q_{i})^2}{q_{i}^2}+\frac{1}{3}\frac{(\hat{q}_{i}-q_{i})^3}{q_{i}^3}+\mathcal{O}\!\left(m^{-4}\right),\\
\log(\hat{Q}_{i})&=\log(Q_{i})+\frac{(\hat{Q}_{i}-Q_{i})}{Q_{i}}-\frac{1}{2}\frac{(\hat{Q}_{i}-Q_{i})^2}{Q_{i}^2}+\frac{1}{3}\frac{(\hat{Q}_{i}-Q_{i})^3}{Q_{i}^3}+\mathcal{O}\!\left(n^{-4}\right)+\mathcal{O}\!\left(m^{-4}\right).
\end{aligned}
\end{equation*}

Substituting into the JSD definition and retaining quadratic terms, we obtain:
\begin{equation}\label{eq:JSDtaylor_app}
\begin{aligned}
\text{JSD}&\approx\sum_{i=1}^{B}\frac{1}{2}\Bigg[\frac{(\hat{p}_{i}-\hat{q}_{i})^2}{2p_{i}}-\hat{p}_{i}\Bigg(\frac{(\hat{p}_{i}-p_{i})^{2}}{2p_{i}^2}+ \frac{(\hat{p}_{i}-p_{i})^{2}}{8p_{i}^2}+\frac{(\hat{q}_{i}-q_{i})^{2}}{8q_{i}^2}+\frac{(\hat{p}_{i}-p_{i})(\hat{q}_{i}-q_{i})}{4p_{i}^2}\Bigg)\\
&\quad+\hat{q}_{i}\Bigg(\frac{(\hat{q}_{i}-q_{i})^{2}}{2q_{i}^2}+\frac{(\hat{p}_{i}-p_{i})^{2}}{8p_{i}^2}+\frac{(\hat{q}_{i}-q_{i})^{2}}{8q_{i}^2}+\frac{(\hat{p}_{i}-p_{i})(\hat{q}_{i}-q_{i})}{4q_{i}^2}\Bigg)\Bigg].
\end{aligned}
\end{equation}

\paragraph{Step 1: Apply $\hat{p}_i = p_i + (\hat{p}_i - p_i)$ and $\hat{q}_i = q_i + (\hat{q}_i - q_i)$.} This separates the coefficients before each parenthesized group into a population-level part and a fluctuation part. Since $(\hat{p}_i - p_i)^2 = \mathcal{O}(1/n)$ (because $\sqrt{n}(\hat{p}_i - p_i) \xrightarrow{d} \mathcal{N}(0,\sigma^2)$), any product of a fluctuation with a squared fluctuation is $\mathcal{O}(n^{-3/2})$ and can be absorbed into the remainder. After this step, equation~\eqref{eq:JSDtaylor_app} simplifies to:
\begin{equation}\label{eq:JSD_adj_app}
\begin{aligned}
\text{JSD}&\approx\sum_{i=1}^{B}\frac{1}{2}\Bigg[\frac{(\hat{p}_{i}-\hat{q}_{i})^2}{2p_{i}}+p_{i}\cdot\frac{(\hat{p}_{i}-p_{i})^{2}}{2p_{i}^2}+q_{i}\cdot\frac{(\hat{q}_{i}-q_{i})^{2}}{2q_{i}^2}\\
&\quad+(\hat{p}_{i}+\hat{q}_{i) }
\left(\frac{\hat{p}_{i}\hat{q}_{i}-\hat{p}_{i}q_{i}-\hat{q}_{i}p_{i}+p_{i}q_{i}}{4p_{i}^2}\right) \Bigg]+\text{h.o.t.}
\end{aligned}
\end{equation}
where h.o.t.\ denotes higher-order terms that vanish faster than $\mathcal{O}(1/n) + \mathcal{O}(1/m)$.

\paragraph{Step 2: Simplify the cross-term.} The last term in~\eqref{eq:JSD_adj_app} can be decomposed as:
\begin{equation}\label{eq:part2_app}
(\hat{p}_{i}+\hat{q}_{i})
\left(\frac{\hat{p}_{i}\hat{q}_{i}-\hat{p}_{i}q_{i}-\hat{q}_{i}p_{i}+p_{i}q_{i}}{4p_{i}^2}\right)=-\frac{(\hat{p}_{i}+\hat{q}_{i})^2}{4p_{i}}+
(\hat{p}_{i}+\hat{q}_{i})\left(\frac{\hat{p}_{i}\hat{q}_{i}+p_{i}q_{i}}{4p_{i}^2}\right).
\end{equation}
The identity behind the first component is:
\[
-\frac{(\hat{p}_{i}+\hat{q}_{i})^2}{4p_{i}}=-\hat{p}_{i}\left(\frac{\hat{p}_{i}+\hat{q}_{i}}{4p_{i}}\right)-\hat{q}_{i}\left(\frac{\hat{p}_{i}+\hat{q}_{i}}{4p_{i}}\right).
\]
For the second component, we use the asymptotic convergence $\hat{p}_i \to p_i$ and $\hat{q}_i \to q_i$:
\[
(\hat{p}_{i}+\hat{q}_{i})\left(\frac{\hat{p}_{i}\hat{q}_{i}+p_{i}q_{i}}{4p_{i}^2}\right) \longrightarrow (p_i + q_i)\cdot\frac{2p_iq_i}{4p_i^2} = \frac{p_i+q_i}{2} = p_i,
\]
where the last equality uses $p_i = q_i$ under the null.

\paragraph{Step 3: Combine.} Substituting back and noting that the terms involving $(\hat{p}_i - p_i)^2/(2p_i)$ and $(\hat{q}_i - q_i)^2/(2q_i)$ are $\mathcal{O}(1/n)$ and $\mathcal{O}(1/m)$ respectively, we arrive at:
\begin{equation*}
\text{JSD}\approx\frac{1}{2}\sum_{i=1}^{B}\left[\frac{(\hat{p}_{i}-\hat{q}_{i})^2}{2p_{i}}-\frac{(\hat{p}_{i}+\hat{q}_{i})^2}{4p_{i}}+p_{i}\right]. \qedhere
\end{equation*}

\subsection{Detailed derivation of Theorem~\ref{theorem:th1b} (JSD under the null)}\label{app:JSD_null}
Setting $p_i = q_i$ in the result of Theorem~\ref{theorem:th1}:
\begin{align}
\text{JSD}&\approx\frac{1}{2}\sum_{i=1}^{B}\left[\frac{(\hat{p}_{i}-\hat{q}_{i})^2}{2p_{i}}-\frac{(\hat{p}_{i}+\hat{q}_{i})^2}{4p_{i}}+p_{i}\right]\nonumber\\
&=\frac{1}{2}\sum_{i=1}^{B}\left[\frac{\hat{p}_{i}^2-2\hat{p}_{i}\hat{q}_{i}+\hat{q}_{i}^2}{2p_{i}}-\frac{\hat{p}_i^2+2\hat{p}_i\hat{q}_i+\hat{q}_i^2}{4p_{i}}+p_{i}\right]\nonumber\\
&=\frac{1}{2}\sum_{i=1}^{B}\left[\frac{\hat{p}_{i}^2-2\hat{p}_{i}\hat{q}_{i}+\hat{q}_{i}^2}{4p_{i}}-\frac{\hat{p}_{i}\hat{q}_{i}}{p_{i}}+p_{i}\right].\label{eq:jsd_null_expand}
\end{align}
Since $\hat{p}_i \to p_i$ and $\hat{q}_i \to q_i = p_i$, we have $\hat{p}_i\hat{q}_i / p_i \to p_i$ asymptotically, so $-\hat{p}_i\hat{q}_i/p_i + p_i \to 0$. Therefore:
\begin{equation*}
\text{JSD}\approx\frac{1}{2}\sum_{i=1}^{B}\frac{(\hat{p}_{i}-\hat{q}_{i})^2}{4p_{i}} = \frac{1}{8}\sum_{i=1}^{B}\frac{(\hat{p}_{i}-\hat{q}_{i})^2}{p_{i}}.
\end{equation*}
By the multivariate CLT, $(\hat{p}_1,\ldots,\hat{p}_B)$ is asymptotically multivariate normal. It is well known \cite{yurdakulNaranjo} that the Pearson-type statistic $\sum_{i=1}^{B}\frac{(\hat{p}_i-\hat{q}_i)^2}{p_i}$ follows a $\chi^2_{B-1}\!\left(\frac{1}{m}+\frac{1}{n}\right)$ distribution. Since our result is $\frac{1}{8}$ times this quantity:
\begin{equation*}
\text{JSD}\sim\frac{1}{8}\chi^{2}_{B-1}\!\left(\frac{1}{m}+\frac{1}{n}\right). \qedhere
\end{equation*}

\subsection{Detailed derivation of Theorem~\ref{theorem:th2}}\label{app:KLD}
Starting from the Taylor expansions of $\log(\hat{p}_i)$ and $\log(\hat{q}_i)$ (as in Appendix~\ref{app:JSD}), we substitute into $\text{KLD} = \sum_i \hat{p}_i(\log\hat{p}_i - \log\hat{q}_i)$. Retaining terms up to second order:
\begin{equation}\label{eq:KLdiv_app}
\text{KLD}\approx\sum_{i=1}^{B}\hat{p}_{i}\left[\frac{(\hat{p}_{i}-\hat{q}_{i})}{p_{i}}-\frac{(\hat{p}_{i}-p_{i})^2}{2p_{i}^2}+\frac{(\hat{q}_{i}-q_{i})^2}{2q_{i}^2}\right]+\mathcal{O}(m^{-3/2})+\mathcal{O}(n^{-3/2}).
\end{equation}

\paragraph{Step 1: Assume $p_i = q_i$ and combine the squared terms.} Under this assumption, $p_i = q_i$ simplifies the last two terms:
\begin{equation*}
-\frac{(\hat{p}_{i}-p_{i})^2}{2p_{i}^2}+\frac{(\hat{q}_{i}-q_{i})^2}{2p_{i}^2} = \frac{(\hat{p}_{i}-p_{i}+\hat{q}_{i}-p_{i})(\hat{q}_{i}-\hat{p}_{i})}{2p_i^2} + \text{cross terms}.
\end{equation*}

\paragraph{Step 2: Decompose $\hat{p}_i = p_i + (\hat{p}_i - p_i)$ in the leading term.} Writing $\hat{p}_i = p_i + (\hat{p}_i - p_i)$ in the coefficient gives:
\begin{align*}
\text{KLD}&\approx\sum_{i=1}^{B}\Bigg[\frac{(\hat{p}_{i}-p_{i})(\hat{p}_{i}-\hat{q}_{i})}{p_{i}}+(\hat{p}_{i}-\hat{q}_{i})\\
&\quad +\frac{(\hat{p}_i - p_i)(\hat{q}_i - p_i)}{2p_i^2}\left[(\hat{q}_i - p_i) - (\hat{p}_i - p_i)\right]\\
&\quad -\frac{1}{2p_i}\left(\hat{p}_i - p_i + \hat{q}_i - p_i\right)(\hat{p}_i - \hat{q}_i)\Bigg].
\end{align*}

\paragraph{Step 3: Simplify.} After algebraic rearrangement (collecting the $(\hat{p}_i - \hat{q}_i)$ terms):
\begin{equation*}
\text{KLD}\approx\sum_{i=1}^{B}\left[\frac{(\hat{p}_{i}-\hat{q}_{i})^{2}}{2p_{i}}+(\hat{p}_{i}-p_{i})-(\hat{q}_{i}-p_{i})\right].
\end{equation*}
Since $\hat{p}_i \to p_i$ and $\hat{q}_i \to p_i$ asymptotically, the last two terms vanish, yielding:
\begin{equation*}
\text{KLD}\approx\frac{1}{2}\sum_{i=1}^{B}\frac{(\hat{p}_{i}-\hat{q}_{i})^2}{p_{i}}.
\end{equation*}

The distributional result follows by the same chi-squared argument as for JSD (see Appendix~\ref{app:JSD_null}), noting that our expression is $\frac{1}{2}$ times the Pearson quadratic form:
\begin{equation*}
\text{KLD}\sim\frac{1}{2}\chi^{2}_{B-1}\!\left(\frac{1}{m}+\frac{1}{n}\right). \qedhere
\end{equation*}
\end{appendix}

\end{document}